\documentclass{eurocg15}

\title{Special Guards in Chromatic Art Gallery~\footnote{ This paper was presented in 17th March 2015 in EuroCG 2015, Ljubljana, Slovenia.}}

\author{Hamid Hoorfar\thanks{Laboratory of Algorithm and Computational Geometry, Faculty of Mathematics and Computer Science,
        Amirkabir University of Technology, Iran.}
        \and
        Ali Mohades\footnotemark[2]}

\begin{document}
\maketitle

\begin{abstract}
We present two new versions of the chromatic art gallery problem that can improve upper bound of the required colors pretty well. In our version, we employ restricted angle guards so that these modern guards can visit $\alpha$-degree of their surroundings. If $\alpha$ is between 0 and 180 degree, we demonstrate that the strong chromatic guarding number is constant. Then we use orthogonal 90-degree guards for guarding the polygons. We prove that the strong chromatic guarding number with orthogonal guards is the logarithmic order. First, we show that for the special cases of the orthogonal polygon such as snake polygon, staircase polygon and convex polygon, the number of colors is constant. We decompose the polygon into parts so that the number of the conflicted parts is logarithmic and every part is snake. Next, we explain the chromatic art gallery for the orthogonal polygon with guards that have rectangular visibility. We prove that the strong chromatic guarding number for orthogonal polygon with rectangular guards is logarithmic order, too. We use a partitioning for orthogonal polygon such that every part is a mount, then we show that a tight bound for strong chromatic number with rectangular guards is ${\theta}(\log n)$
\end{abstract}

\section{Introduction}
New approach of the art gallery problem was raised with Erickson and LaValle~\cite{lavalle}, which maximized the compatible guards so that for two guards whom their intersection of visibility polygons is not empty must be spent a new color as cost. In the other words, the chromatic art gallery find the minimum number of colors that always sufficient and sometimes necessary for guarding the entire polygon. It is called the chromatic guarding number. Let $\chi_{G}(P)$ denotes the chromatic guarding number of polygon P. We extend this notation so that $\chi_{G}^{\alpha}(P)$ denotes the chromatic guarding number with $\alpha$-degree guards, and  $\chi_{G}^{rec}(P)$ denote the chromatic guarding number with rectangular guards, properly. The motivation of offering this problem was in robot controlling with wireless navigators whom we can set the angle of their ranges. In many cases, these navigators have the 90-degree range, and they are orthogonal corresponding to the environment, because of it, we introduce the chromatic art gallery with orthogonal 90-degree guards in the orthogonal polygons. Erickson and LaValle showed that for a spiral polygon, the chromatic guarding number is at most 2 and for a staircase polygon is at most 3, then they showed that for every positive number k, there exists a polygon with $4k$ vertices such that$\chi_{G}(P_{k})\geq k$ . Also, they showed that for every odd number k there is an orthogonal polygon with $4k^2+10k+10$ vertices such that $\chi_{G}(P_{k})\geq k$ ~\cite{lavalle}. We extend these result with $\alpha$-degree guards so that $\chi_{G}^{\alpha}(P)\leq 2$  and for orthogonal guards $\chi_{G}^{O}(P)\in O(\log n)$ wherever $\chi_{G}^{O}(P)$ denotes the chromatic guarding number with the orthogonal 90-degree guards. After that, we extend these result with rectangular guards so that $\chi_{G}^{rec}(P)\in {\theta}(\log n) $

\section{Basic Definitions}
Let polygon P be a connected simple subset of $\Bbb{R}^2$ with $\partial$P as its boundary. $p\in P$ is visible from $q\in P$  if the segment  pq is a subset of $P$. For every point $p$ in the polygon, $V(p)$ indicates visibility polygon  of p so that $V(p)=\{q\in P\,\, s.t.\,\, p\,  is\, visible \, from\,  q\}$. The guard set $S$ is a finite set of points in the polygon such that $\bigcup_{s\in S}V(s)=P$, every element of $S$ is called guard~\cite{Rourke}. A pair of guards is named incompatible whenever $V(s)\cup V(t)\neq \emptyset$. Let $C(s)$ denotes the minimum colors necessary for coloring the guard set $S$ so that every pair of incompatible guards have different colors.  Furthermore, let $T(p)$ denote the set of all guard sets in $P$, and $\chi_{G}(P)=min_{s\in T(p)}  C(s)$  is called the chromatic guarding number. The chromatic art gallery problem minimizes the chromatic guarding number rather the guarding number~\cite{lavalle2}. Let $\alpha$-guard denotes the guard whom its visual field is $(v,v+\alpha)$ wherever $v$ is an arbitrary angle and $\alpha \in (0,180]$, for instance, 90-guard is a guard with visual field equal to $(v,v+90)$. In addition, let O-guard denotes the orthogonal 90-guard. Suppose $\chi _{G}^{\alpha} (p)$  denotes the chromatic guarding number with the $\alpha$-guards, and we extend the notations so that $\chi _{G}^O (p)$  indicates the chromatic guarding number with the O-guards as well. Also:\\
$V^{\alpha} (g)=\{p\in P\mid q\,\, is\,\, visible\,\, from\,\, \alpha-guard\,\, g\}$\\
$V^{O} (g) =\{q\in P\mid q\,\, is\,\, visible\,\, from\,\, o-guard\,\, g\}$\\
In the orthogonal polygon, a horizontal (vertical) edge $e_h$ ($e_v$) that its two end points are reflex vertices is called h-cut edge (v-cut edge)~\cite{Ghosh}. For any two points p and q in a rectilinear polygon P, if the aligned rectangle with p and q as opposite corners lies totally inside P, then p and q are called rectangularly visible~\cite{Ghosh}. Also for any point p in polygon P, $V^{rec} (p)$ denotes the rectangular visibility area of P, such that:\\$V^{rec} (p) =\{q\in P\vert p\, is\, rectangularly\, visible\, from\, q\}$\\ Every guard can see around it rectangularly is represented by the r-guard. Assume guard set S that is a finite set of r-guards in the polygon so that $\bigcup_{s\in S}V^{rec}(s)=P$ then we call that P is r-visible from S. we extend the use of the notations to r-guards as $\chi_{G}^{rec} (P)=\min_{S\in T(p)}{C(S)}$ whenever S is an r-guard set.
And assume these notations:\\
$V^{\bot} (p)=\{ q\in P\vert q\, is\, orthogonally\, visible\, from\, point\, p\}$\\
$V^{\bot} (e)=\{ q\in P\vert q\, is\, orthogonally\, visible\, from\, segment\, e\}$

\begin{obs}
Every orthogonal polygon that has no h-cut edge (v-cut edge) is y-monotone (x-monotone) polygon. As the figure ~\ref{fi:fig1}, it is showed the orthogonal x-monotone polygon has no v-cut edges.

\begin{figure}
	\includegraphics[width=\columnwidth]{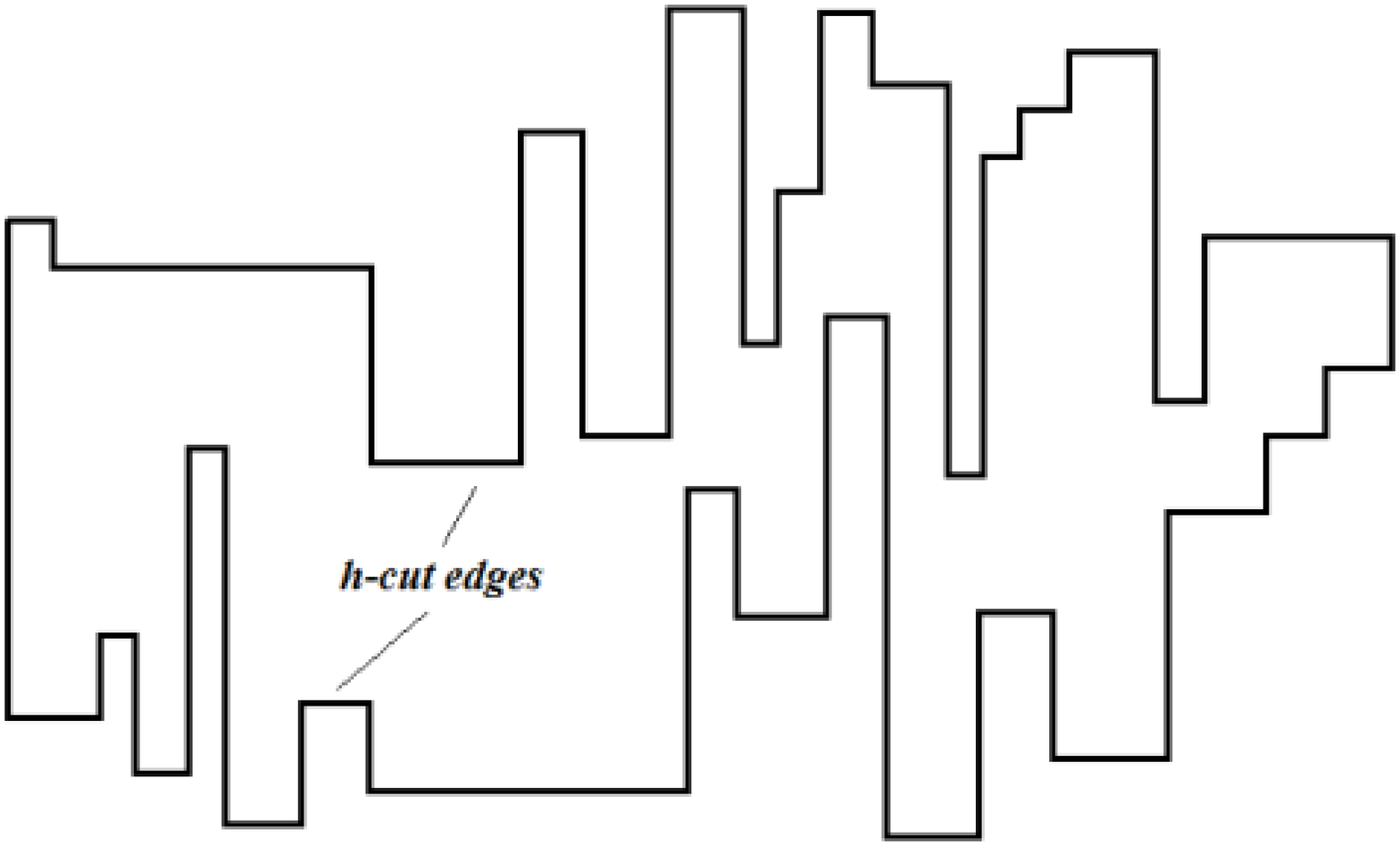}
	\caption{An orthogonal polygon that has no h-cut edge.}
	\label{fi:fig1}
\end{figure}

\end{obs}

In the following we definite a new special polygon, we call it, the snake polygon. 
\begin{defini}
A polygon P is called snake polygon if:
\\ \hspace{1.1em} 1. P is an orthogonal x-monotone (y-monotone).
\\ \hspace{1.1em} 2. Extend every h-cut edge decompose polygon into three sub polygons, exactly.
\\ \hspace{1.1em} 3. One of its sub polygon must be xy-monotone.
see figure ~\ref{fi:fig2}. 
\end{defini}
As the sample, every staircase polygon is a snake polygon; similarly, every xy-monotone polygon is a snake polygon as well.
\begin{figure}
	\includegraphics[width=\columnwidth]{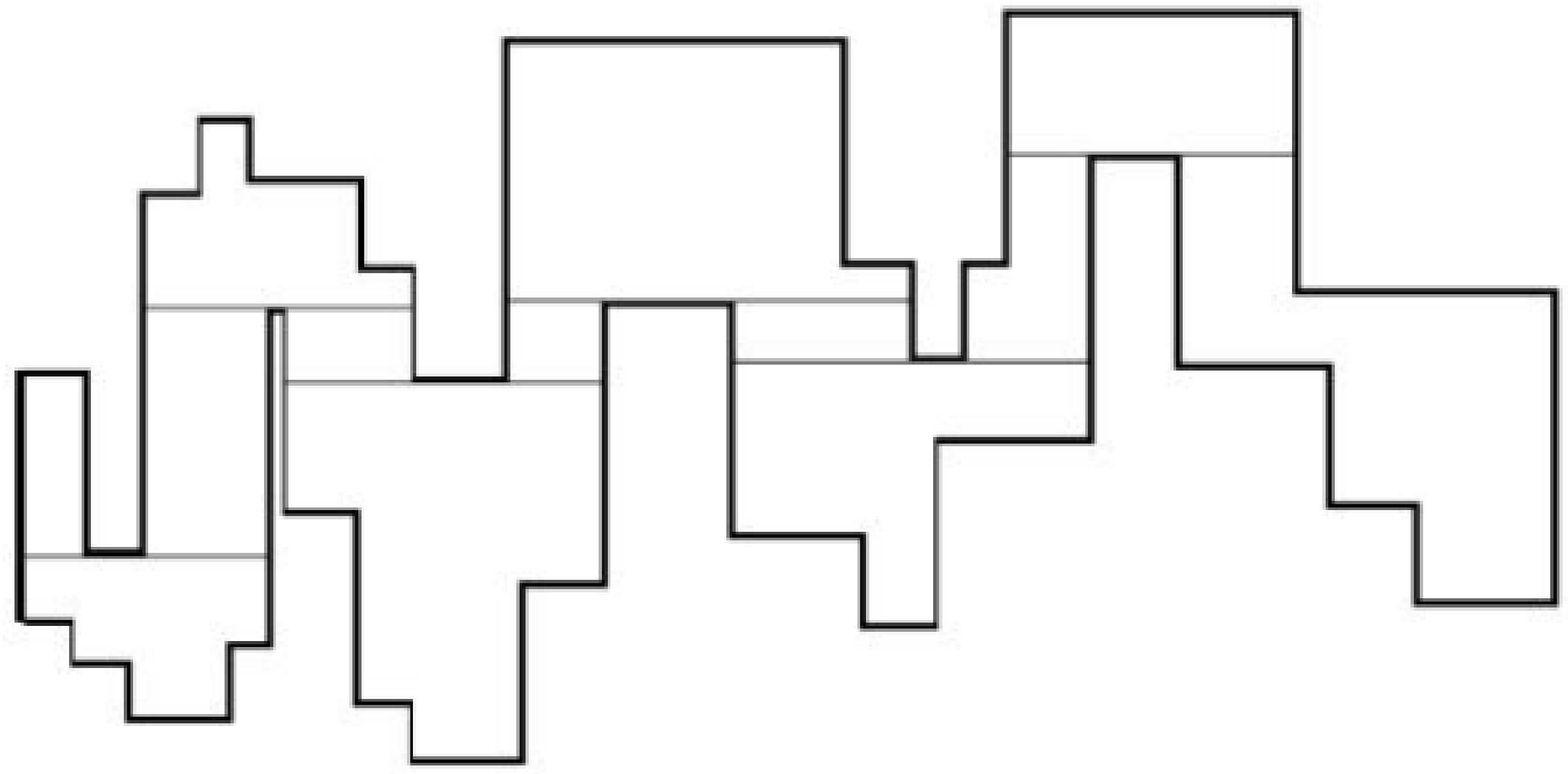}
	\caption{A snake polygon}
	\label{fi:fig2}
\end{figure}
We definite two special polygons, mounts polygon and mount polygon as following.
\begin{defini}
Polygon P is called mounts polygon if it has at least an edge $e\in P$  such that $V^{\bot}(e) =P$ , nominate e, base edge.
\end{defini}

\begin{defini}
The mounts polygon P is named mount polygon if P is monotone corresponding to a line perpendicular to its base edge.
\end{defini}

\section{A tight bound on the chromatic guarding number of the general polygon with $\alpha$-guards}

Consider a simple polygon $P$. Select an arbitrary edge $e_{1}$ then place a 180-guard $g_1$ on it.  Process $V^{180} (g_{1} )$~\cite{Elgindy}, every connected part of its boundary, which is not belonged $\partial$P is called window. Continually place 180-guards on the created windows of previous step until the entire polygon is covered, see figure ~\ref{fi:fig3}

\begin{figure}
	\includegraphics[width=\columnwidth]{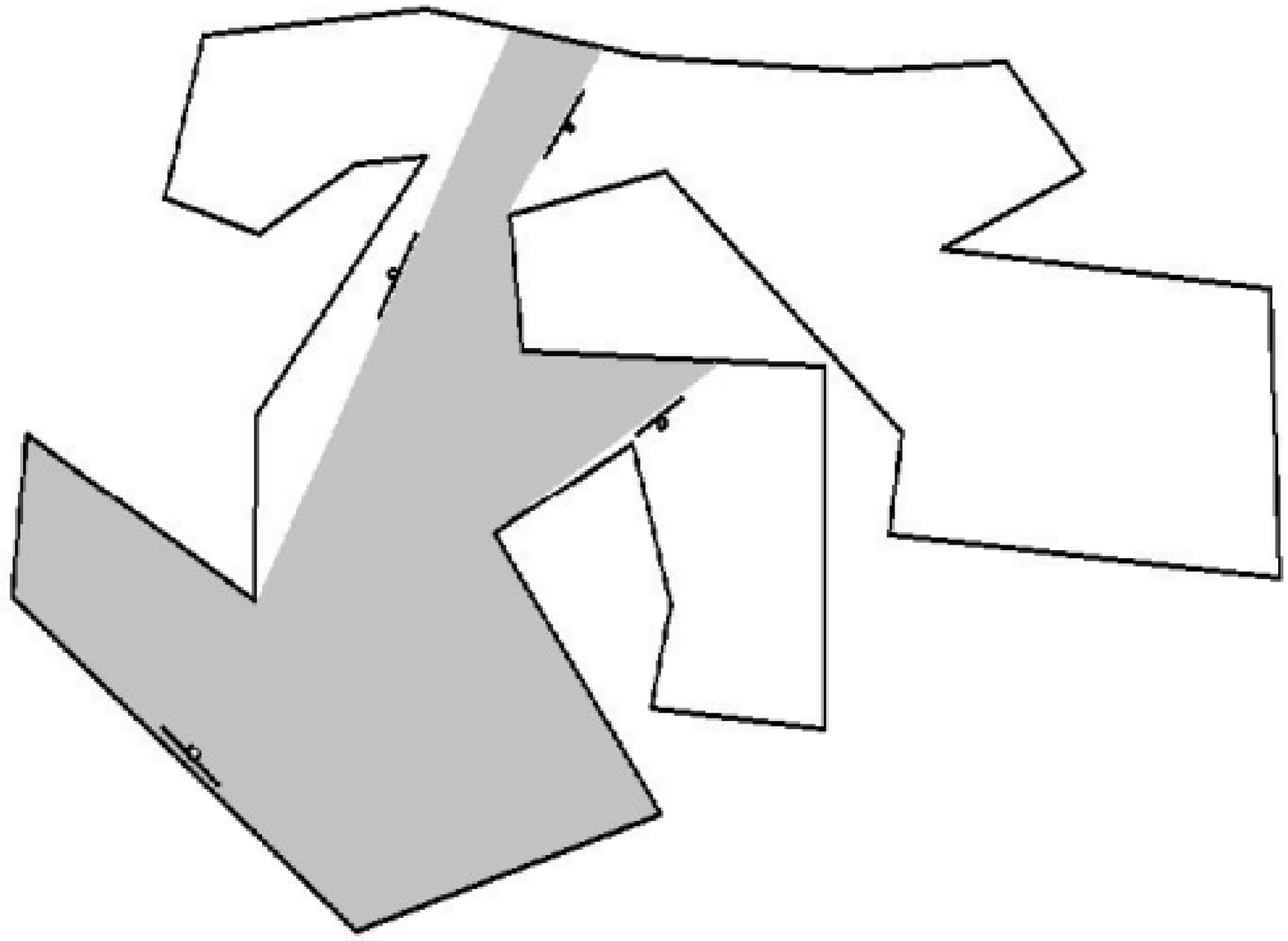}
	\caption{180-degrees guard}
	\label{fi:fig3}
\end{figure}

We demonstrated that in the general polygon with 180-degree guards the chromatic guarding number is 1, i.e.$\chi_{G}^{180} (p)=1$. Now, we can replace every 180-guard with at most two colors of $\alpha$-guards.  We want $\lfloor$$180\over\alpha$$\rfloor$ guards in same color and perhaps one guard in different color wherever $180\over\alpha$ is not an integer such that it is shown in figure ~\ref{fi:fig4}. Hence, if $0<\alpha \leq180$, then $\chi_{G}^{\alpha} (p)\leq2$.

\begin{figure}
	\includegraphics[width=\columnwidth]{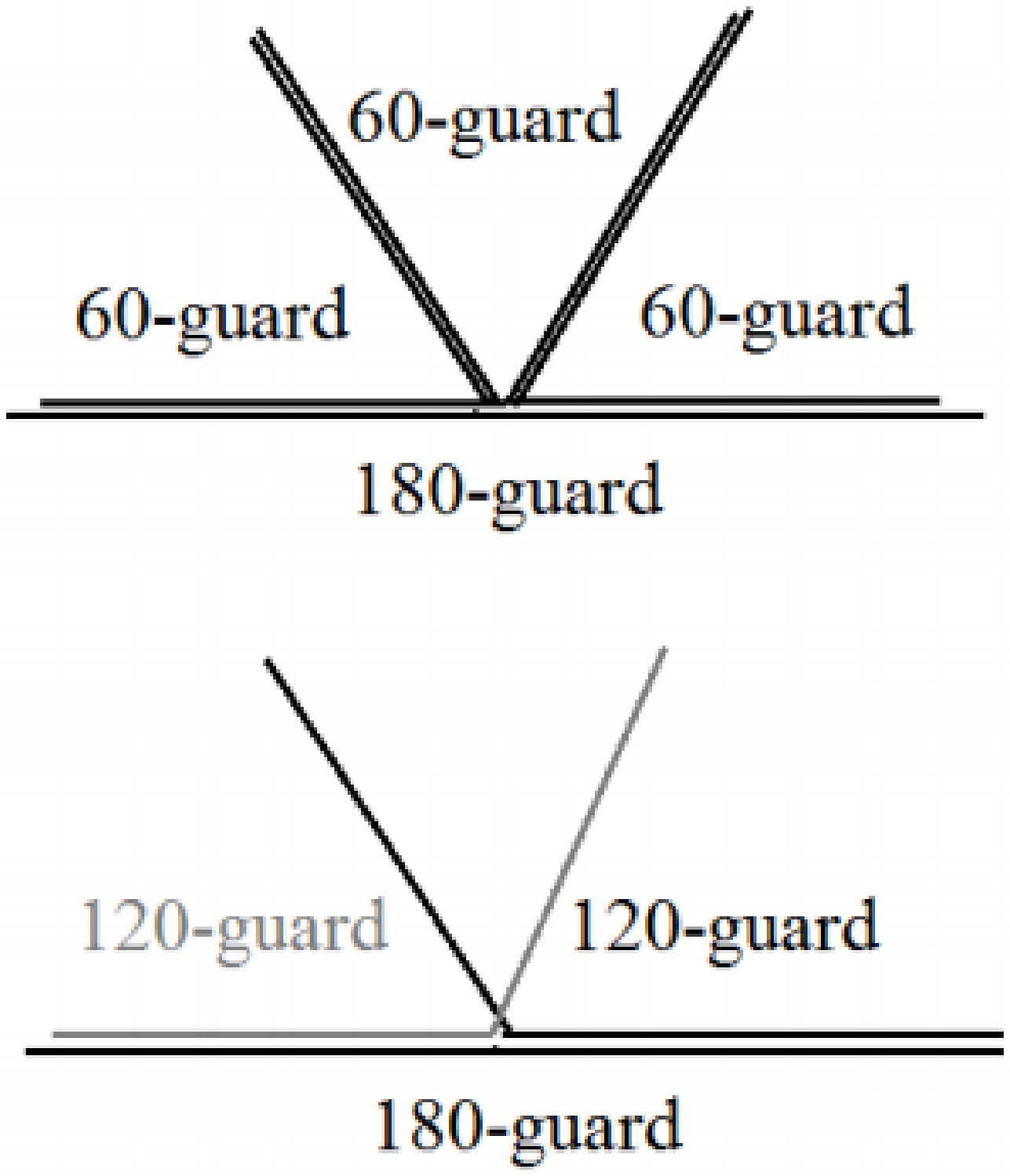}
	\caption{Replace a 180-degrees guard with $\alpha$-guards}
	\label{fi:fig4}
\end{figure}

\section{Some special polygons with O-guards}

Suppose snake polygon P, we present a partitioning so that its chromatic guarding number will be constant. Extend all h-cut edges in the snake polygon, because of this, the polygon decomposes to sub polygons such that all of them are staircase or mount polygons. We can guard all these staircase parts with one color of the O-guard independently as any parts are conflict-free from others, see figure ~\ref{fi:fig5}. The green parts are the mount polygons, and the blue ones are staircase polygons.

\begin{figure}
	\includegraphics[width=\columnwidth]{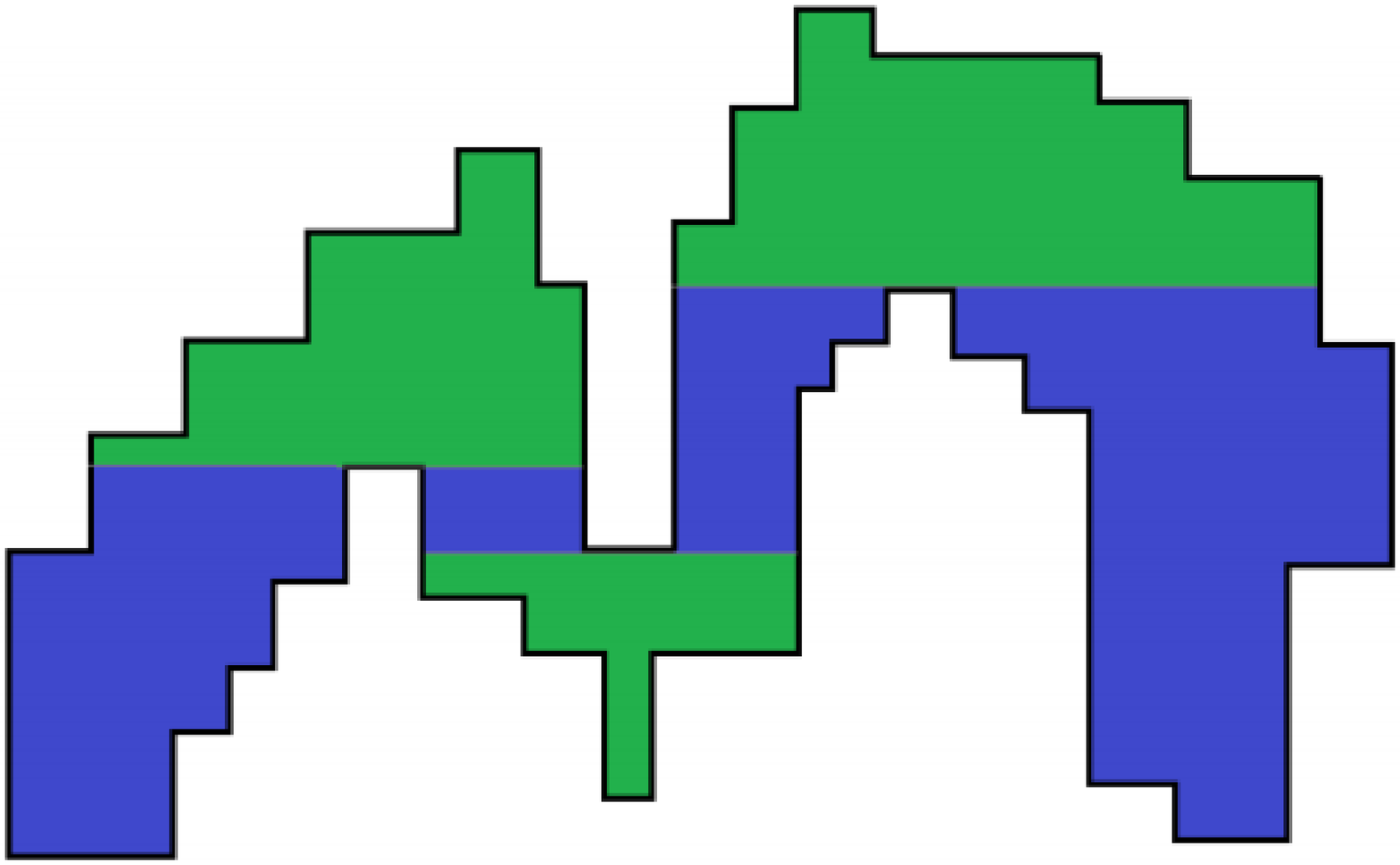}
	\caption{O-guarding a snake polygon}
	\label{fi:fig5}
\end{figure}

\begin{obs}
For every staircase polygon $P$ with the condition that the bottom edge is not seen, $\chi_{G}^{O} (P)=1$.
\end{obs}

\begin{obs}
For every mount polygon $P$ with this condition that the bottom edge is not seen, $\chi_{G}^{O} (P)=1$, figure ~\ref{fi:fig6}  shows, it is clear.

\begin{figure}
	\includegraphics[width=\columnwidth]{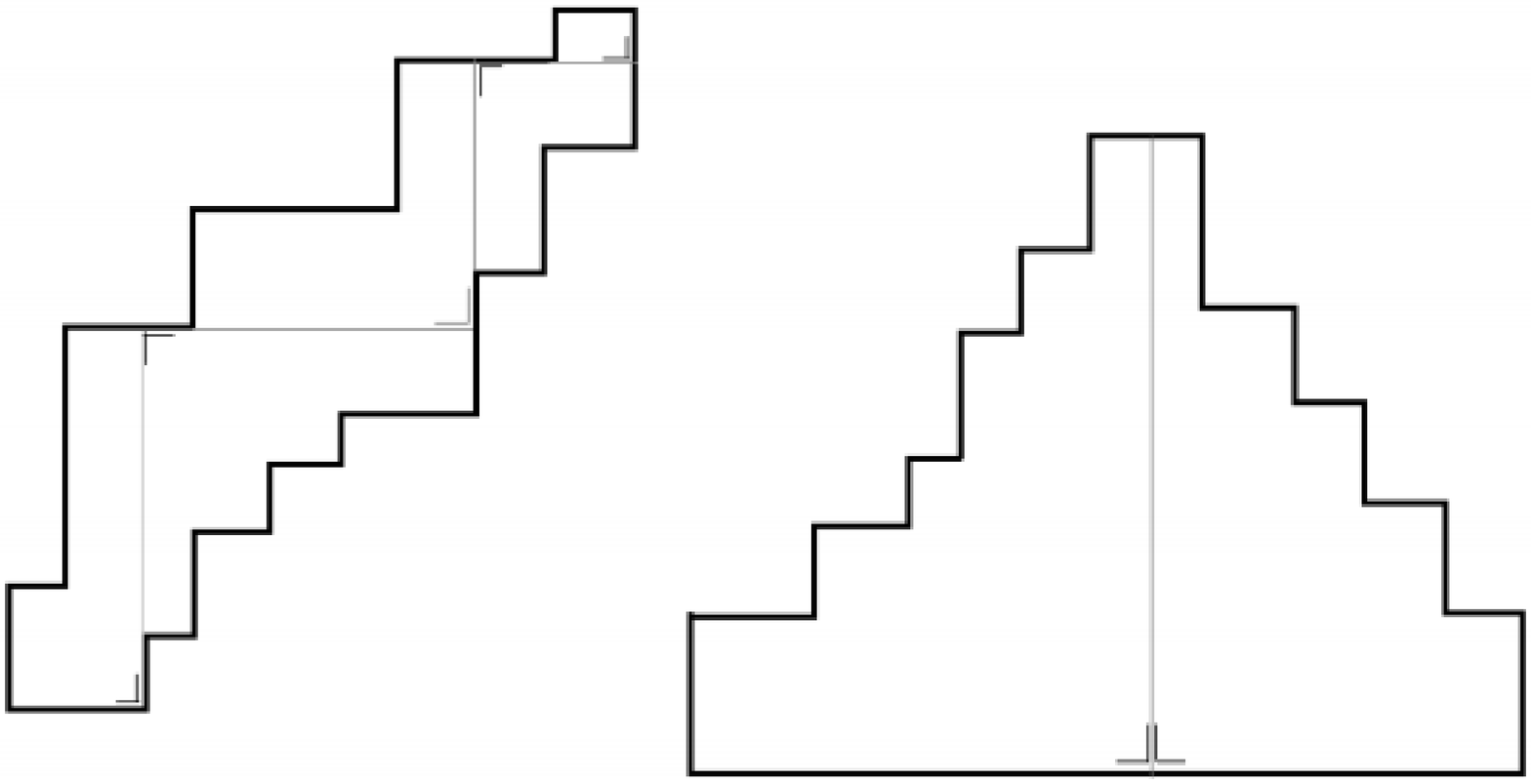}
	\caption{The staircase polygon and mount polygon can be guarding with one color of O-guards}
	\label{fi:fig6}
\end{figure}

\end{obs}

If the bottom edge is not seen, then the visibility polygon of the O-guards in two different parts will not be incompatible. We place the green O-guards in the green areas and blue ones in the blue areas. Therefore, we can cover the entire snake polygon with at most two colors. So the chromatic guarding number of the snake polygon is less than or equal to 2.

\begin{obs}
The chromatic guarding number of the snake polygon is at most 2.
\end{obs}

\section{A tight upper bound on the chromatic guarding number of the orthogonal polygon with O-guards}

In this section, we present a partitioning for the orthogonal polygon such that every part is a snake polygon. We construct two partitions in two orientations both horizontally and vertically. Call horizontal (vertical) partitioning h-cutting (v-cutting), then we nominate duality graph of h-cutting (v-cutting), h-tree (v-tree). Extend all h-cut (v-cut) edges in the polygon until intersect the boundary, start from the lowest (leftmost) part, the duality node corresponding to it must be root, then draw a directed edge from the root to all its neighbors, continue it recursively until the h-tree (v-tree) is built. See figure~\ref{fi:fig7} .

We modify any paths in the h-tree, so that duality of them will be snake polygons. For this purpose, we cut the v-cut edges occur in the path. We know that these removed parts will cover with guarding in the other orientation partitioning, certainly. See figure~\ref{fi:fig8}.

\begin{figure}
	\includegraphics[width=\columnwidth]{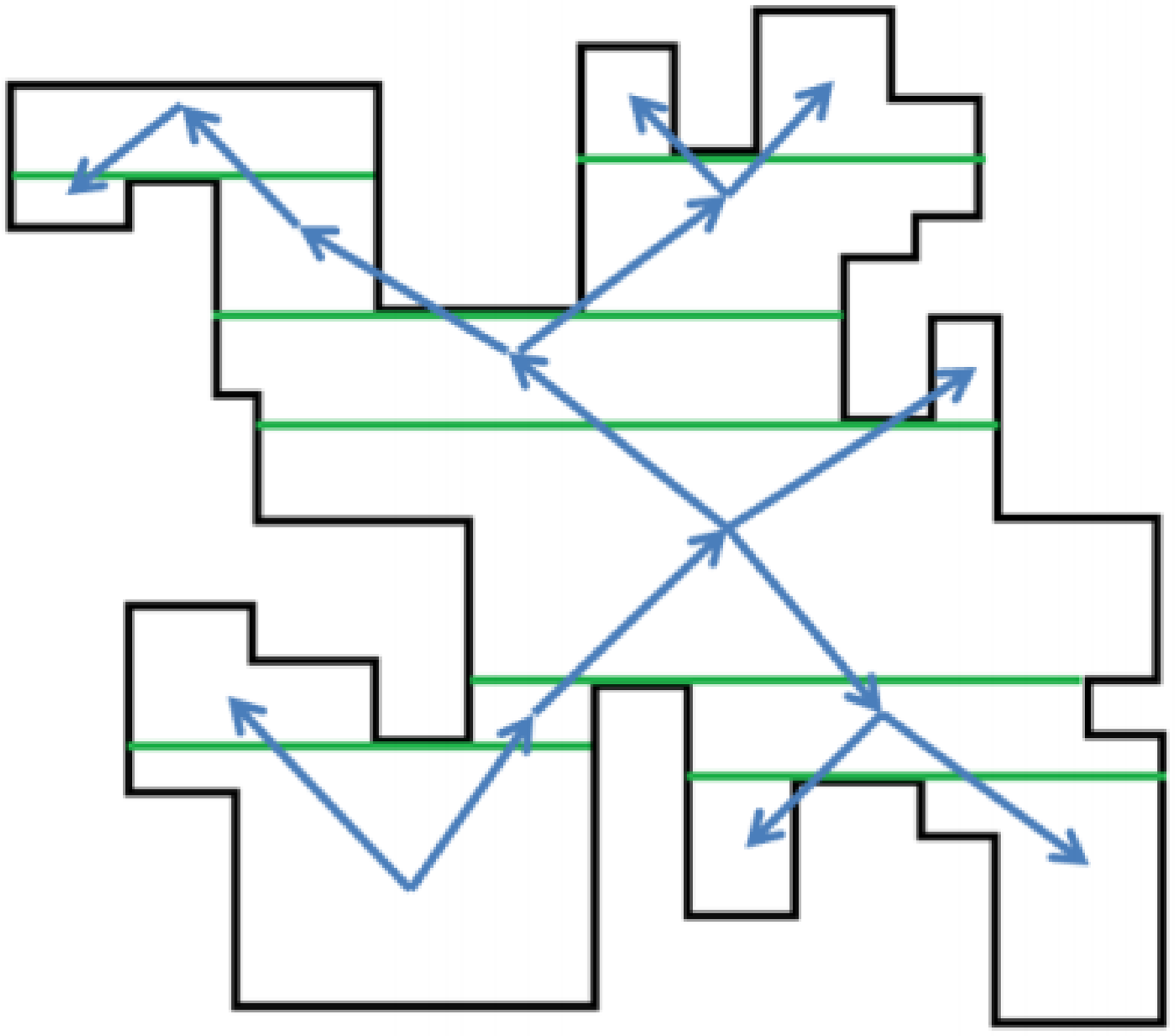}
	\caption{Every path in the duality graph decomposes to a snake polygon.}
	\label{fi:fig7}
\end{figure}

Similarly, consider this process for v-tree. If we find the chromatic guarding number to cover the h-tree, double of this number will be the chromatic guarding number for the entire orthogonal polygon. We know that every modified path in the duality tree has the chromatic guarding number at most 2 with the condition that the bottom edge is not seen. Follow this algorithm:\\
1. Find the path from the root to the leaf in every remained tree so that with removing it, the tree decomposes into at least two sub trees which their size are at most half of the total tree.\\
2. Guard all paths from previous step with two new colors.\\
3. Remove the path from the total tree and remain the sub tree(s).\\
4. If all remained sub trees are not empty, then go to step 1.\\

\begin{lemma}
\label{le:great200}
The number of iterations of the algorithm is $O(\log n)$.
\end{lemma}
\begin{proof}
In iteration the size of the problem will be at most half of the pervious step then we can write:\\
$f(n)\leq f($$n\over{2}$$)+2 \,\, \Rightarrow\,\, f(n)\in O(\log n)$.
$f(n)$ means the complexity of problem.
\end{proof}

\begin{figure}
	\includegraphics[width=\columnwidth]{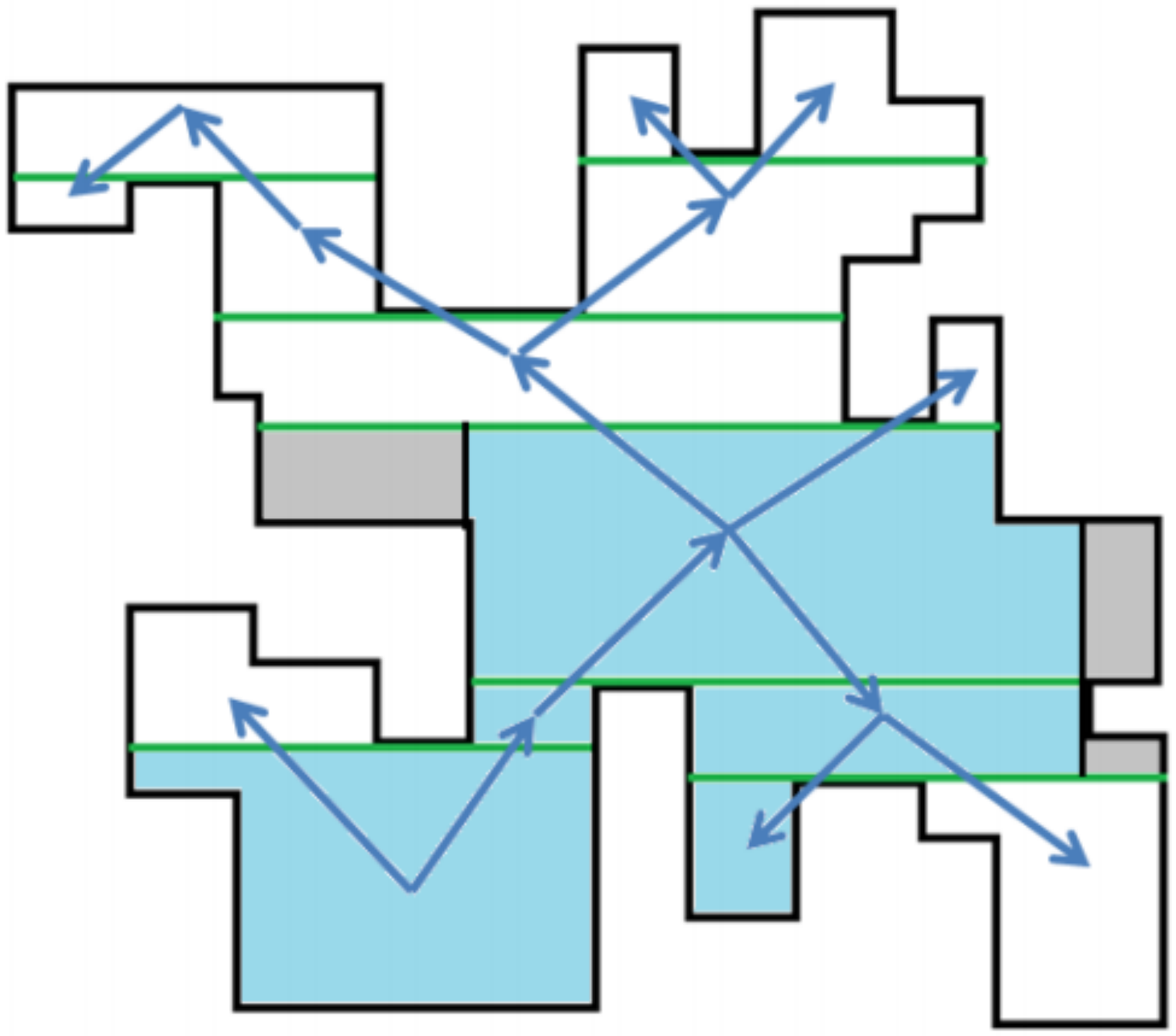}
	\caption{Partitioning of an orthogonal polygon.}
	\label{fi:fig8}
\end{figure}

\begin{theorem}
\label{le:great}
The chromatic guarding number of the orthogonal polygon is the logarithmic order.
\end{theorem}
\begin{proof}
Using Lemma~\ref{le:great200}, we know that for guarding the h-cutting, we need $O(\log n)$ colors of guards, by this guarding, some parts of the polygon is not covered, because of the modified paths, that cut in the cross orientation. Nevertheless, we are guarding these remained parts during guarding v-cutting, completely. For it, we pay $O(\log n)$ colors of guards as cost. As a result, the chromatic guarding number for the orthogonal polygon belongs to $O(\log n)$.
\end{proof}

\section{An upper bound on the chromatic guarding number of the general polygon with O-guards}

We present an algorithm for the general polygon such that we remove an orthogonal polygon from it. By it, all remained parts are spiral polygons or triangles so that we can guard them with at most two colors compatibly. Follow this algorithm:
\\1. Draw from every vertex in the polygon a horizontal line.
\\2. From every intersection of lines and polygon, draw a vertical line.
\\3. From all rectangles so that occurs in the polygon completely, select it in the target orthogonal polygon. 
\\4. Consider remained parts, we can guard all disjoint polygon compatibly.\\

\begin{obs}
Every remained part is the spiral polygon so that reflex chain of it, is orthogonal. (Triangle is a special type of the spiral polygon).
\end{obs}

\begin{obs}
The chromatic guarding number for the spiral polygon so that the reflex chain is not seen with o-guard is at most 2.
\end{obs}

The reason of this condition that the reflex chain doesn't must be seen, is  the spiral parts can guard independently.

\begin{lemma}
\label{le:great2}
The chromatic guarding number for the general polygon is $O(\log n)$.
\end{lemma}
\begin{proof}
Decompose a general polygon into an orthogonal polygon that its vertices can be $O(n^2)$, nevertheless the $\chi_{G}^{O}(P)\in O(\log⁡ n)$ , for remained spiral polygon  $\chi_{G}^{O}(P)\leq 2$,  therefore, for the general polygon P, we have:\,\,
$\chi_{G}^{O}(P)\in O(\log⁡ n)$

\end{proof}

\section{A tight bound on the rectangular chromatic guarding number of the orthogonal polygon}
\begin{obs}
The rectangular chromatic guarding number of the mounts polygon belongs to $\theta(\log n)$.
\end{obs}
We found that the rectangular chromatic guarding number of mounts polygon is the logarithmic order. At present, we want to extend this result for all orthogonal polygons as well. It does with partitioning the polygon into the mounts polygons. Find the lowest edge of the polygon, nominate e and process $V^{\bot} (e)$. Suppose the boundary of $V^{\bot} (e)$, every connected part of the boundary which is not part of the boundary of the total polygon is called window.

\begin{obs}
The connected window in $V^{\bot} (e)$ is line segment.
\end{obs}

Now, repeat finding $V^{\bot} (e)$ for every obtained window since the entire polygon is supported. See figure~\ref{fi:fig0006}.

\begin{figure}
	\includegraphics[width=\columnwidth]{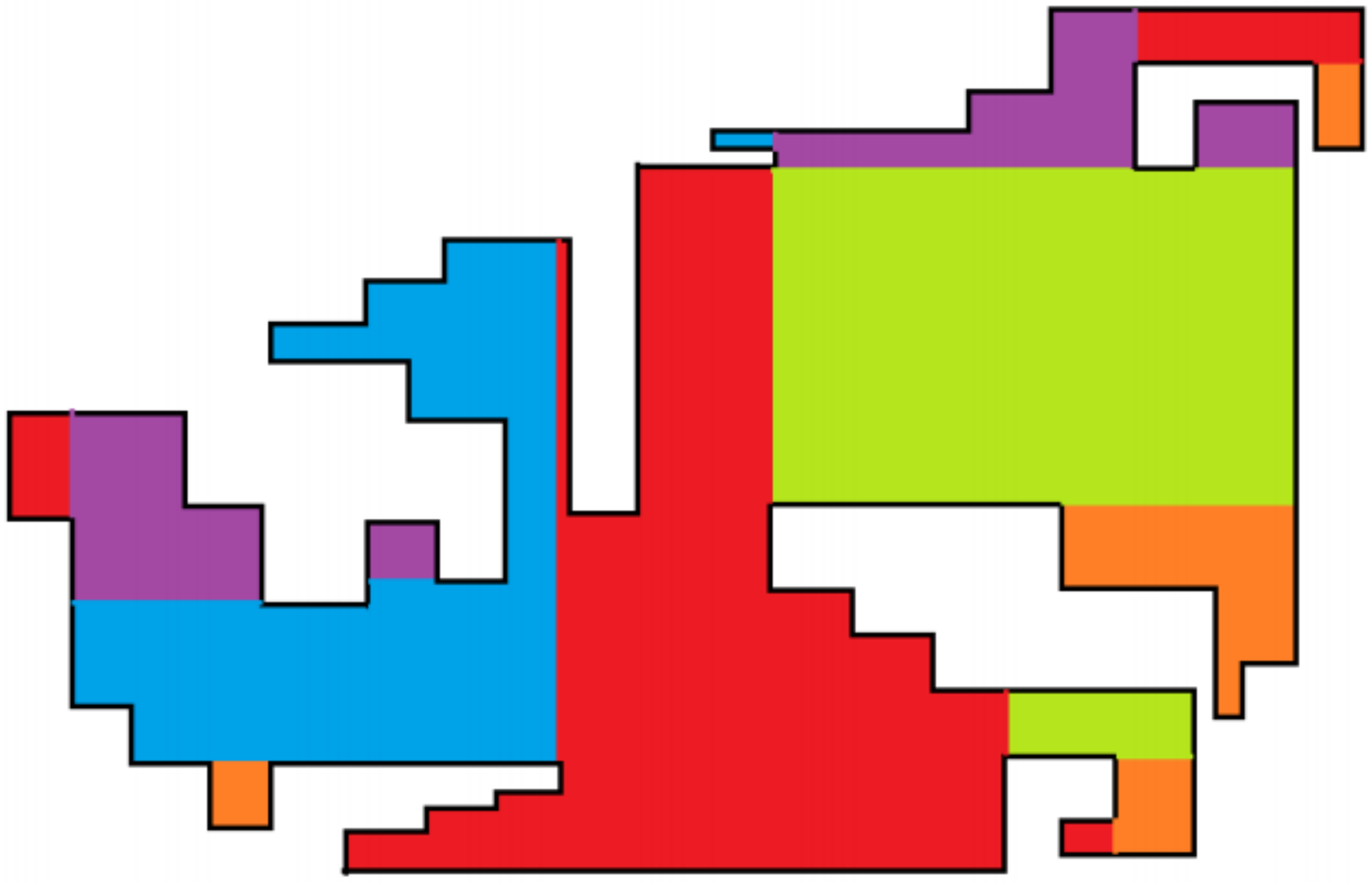}
	\caption{Partitioning the orthogonal polygon to the mounts polygons.}
	\label{fi:fig0006}
\end{figure}

Now, we demonstrate that if we can cover the mounts polygon with $\log n$ colors, then we will can cover all parts of the partitioning with $5\log n$ colors, in other words:\,
$\chi_G^{rec} (P_{orthogonal})\in \theta(\log⁡ n)$.\\
If a mounts polygon part has the window in its boundary, then it can have incompatibility with other parts. The incompatibility area is a rectangle such that its width is equal to width of the window. Start from first part, color it with the first color and color all parts that immediately left of it (call Left-Parts) with the second color and all parts that immediately right of it (call Right-Parts) with the third color. So, color all parts top of Left-Parts and Right-Parts with 4th color and all parts bottom of Left-Parts and Right-Parts with 5th color. Any parts in the same color want the same log n colors of r-guards. By this strategy, we can color the remained parts with the same five colors so that showed in figure~\ref{fi:fig0006}.

\begin{obs}
The same colored parts are compatible.
\end{obs}

\begin{theorem}
\label{le:great6}
The rectangular chromatic guarding number that always sufficient and sometimes necessary for covering an orthogonal polygon is O(log ⁡n).
\end{theorem}
\begin{proof}
We found that there is an orthogonal polygon which is needed $\theta(\log n)$  as its rectangular chromatic guarding number, and using above method are shown that $\theta(\log n)$ is sufficient for all orthogonal polygons, so the proof is completed.
\end{proof}

\section{Conclusion}

In this paper, we explained the chromatic art gallery on the simple polygon with 90-guards that its visibility field of them is one of the (0, 90), (90,180), (180,270) or (270,360). This type of guards named O-guard. In many cases in real-world, the guards can see a limited angle of its around. This motivated that we define the new version of the chromatic art gallery problem so that increase the conflict-free guards from lower bound $n\over 4$ to upper bound $O(\log n)$. Then, we explained the chromatic art gallery for the orthogonal polygon with a special type of guards are called r-guards. This type of guards has r-visibility such that any two points $p$ and $q$ in a rectilinear polygon $P$ are called r-visible, if the aligned rectangle with $p$ and $q$ as opposite corners lies totally inside $P$.
The applications of this problem are in mobile network coverage, remote controlling, wireless communications and etc. this research showed that by decreasing the angle field of Transmitters, we can save the frequency of the waves for the communications.

%---------------------------- Bibliography -------------------------------

% Please add the contents of the .bbl file

\small
\bibliographystyle{abbrv}

\begin{thebibliography}{99}

\bibitem{Rourke}
J. ~O’Rourke.
\newblock {\em  Art Gallery Theorems and Algorithms.}
\newblock Oxford University Press, Cambridge, 1987, UK.

\bibitem{lavalle}
L. H.~Erickson. S. M. ~ LaValle.
\newblock {\em An art gallery approach to ensuring that landmarks are distinguishable.}
\newblock The Journal of Science and Systems, 2011,  Los Angeles, USA.


\bibitem{lavalle2}
L. H.~Erickson. S. M. ~ LaValle.
\newblock {\em A chromatic art gallery problem.}
\newblock Technical report, 2010,  illinois, USA.

\bibitem{Ghosh}
S. K.~Ghosh.
\newblock {\em   Visibility Algorithms in the plane.}
\newblock  Cambridge University Press, 2007, UK.


\bibitem{Elgindy}
H.~ ElGindy. D. ~Avis.
\newblock {\em A Linear algorithm for computing the visibility polygon from a point.}
\newblock The journal of Algorithms, 1987.



\end{thebibliography}

\end{document}